\documentclass[twoside]{article}
\usepackage[a4paper]{geometry}
\usepackage[latin1]{inputenc} 
\usepackage[T1]{fontenc} 
\usepackage{RR/RR}
\usepackage{hyperref}

\usepackage{times}

\usepackage[latin1]{inputenc} \usepackage[T1]{fontenc} \usepackage[english]{babel}
\usepackage{url,xspace,color,subfigure} \usepackage{pdfpages} \usepackage{placeins}

\usepackage{amsfonts,amsmath,amssymb,amsthm,amstext,latexsym,paralist}

\graphicspath{{figs/},{RR/}}

\newtheorem{theorem}{Theorem} 

\newcommand{\ema}[1]{\ensuremath{#1}\xspace}

  \newcommand{\ie}{{\it i.e.}\xspace}

\newcommand{\MM}{\ema{\mathcal M}} \newcommand{\PP}{\ema{\mathcal P}}
\newcommand{\alloc}{\ema{\mathcal A}} 
\newcommand{\fail}{\ema{\textsc{Fail}}}
\newcommand{\alive}{\ema{\textsc{Alive}}}
\newcommand{\alivese}[1][\null]{\ema{\textsc{AliveInst}_{#1}}}

\newcommand{\nbm}{\ema{m}}

\newcommand{\machine}{\ema{\mathcal M}} \newcommand{\mach}[1]{\ema{\machine_{#1}}}
\newcommand{\service}{\ema{\mathcal S}} 
\newcommand{\capacity}{\ema{\textsc{Capa}}} 
\newcommand{\demand}{\ema{\textsc{Dem}}} 
\newcommand{\relia}{\ema{\textsc{Rel}}} 

\newcommand{\all}[2]{\ema{\alloc_{#1,#2}}} \newcommand{\allhom}[1]{\ema{\alloc_{#1}}}

\newcommand{\expe}[1]{\ema{\mathbb{E}\left( #1 \right)}} \newcommand{\pro}{\ema{\mathbb{P}}}

\newcommand{\pf}{\ema{\pro_{\mathit{fail}}}}

\newcommand{\logn}[1]{\ema{\mathit{ln}\left( #1 \right)}}

 \newcommand{\lrd}[1][\null]{K_{#1}}

\newcommand{\energy}{\ema{E}} \newcommand{\estat}{\ema{\energy_{\mathrm{stat}}}}
\newcommand{\edyn}{\ema{\energy_{\mathrm{dyn}}}} \newcommand{\pen}{\ema{\alpha}}
 
\newcommand{\ennl}[1]{\ifx&#1&\ema{\energy^\ilo}\else \ema{\energy^\ilo(#1)}\fi}
\newcommand{\dennl}[1]{\ifx&#1&\ema{\left(\energy^\ilo\right)'}\else
  \ema{\left(\energy^\ilo\right)'(#1)}\fi} \newcommand{\ennu}[1]{\ifx&#1&\ema{\energy^\iup}\else
  \ema{\energy^\iup(#1)}\fi} \newcommand{\dennu}[1]{\ifx&#1&\ema{\left(\energy^\iup\right)'}\else
  \ema{\left(\energy^\iup\right)'(#1)}\fi}

\newcommand{\iup}{\mathrm{(up)}} \newcommand{\ilo}{\mathrm{(low)}}

\newcommand{\xmi}{\ema{x_{\mathrm{min}}}} \newcommand{\xmiu}{\ema{\xmi^{\iup}}}
\newcommand{\xmil}{\ema{x_{\mathrm{min}}^{\ilo}}}

\newcommand{\son}{\ema{\textsc{On}}}

\newcommand{\minrep}{\ema{\textsc{Min-Replication}}}
\newcommand{\minenns}{\ema{\textsc{Min-Energy-No-Shutdown}}}
\newcommand{\minen}{\ema{\textsc{Min-Energy}}}

\newcommand{\hhomo}{\textit{Homogeneous}\xspace} \newcommand{\hstep}{\textit{Step}\xspace}

\newcommand{\glb}{\textbf{lower.bound}\xspace} \newcommand{\gth}{\textbf{theo.homo}\xspace}
\newcommand{\gbh}{\textbf{best.homo}\xspace} \newcommand{\gbs}{\textbf{best.step}\xspace}

\newcommand{\inter}[2][1]{\{ #1 , \dots , #2 \}}

\RRNo{8241}
\RRdate{February 2013}
\RRauthor{
 Olivier Beaumont\thanks[all]{Inria, University of Bordeaux}
\and
Philippe Duchon\thanksref{all}
\and
Paul Renaud-Goud\thanksref{all}
}
\authorhead{Beaumont \& Duchon \& Renaud-Goud}
\RRtitle{Algorithmes d'approximation sur la minimisation d'énergie\\
pour l'allocation de services dans un Cloud\\ sous contraintes de fiabilités}
\RRetitle{Approximation Algorithms for Energy Minimization in Cloud Service Allocation\\ under Reliability Constraints}
\titlehead{Energy Minimization in Cloud Service Allocation}
\RRresume{
Nous considérons un problème d'allocation de services dans des \textit{Clouds}. Les resources de
calcul sont caractérisées par une probabilité de panne, et une contrainte de capacité, qui peut
être ajustée grâce à la technique dite de \textit{Dynamic Voltage and Frequency Scaling} (DVFS).
Il existe un contrat entre le fournisseur et le client, le fournisseur assurant au client qu'un
certain nombre d'instances du service du client sera toujours en train de s'exécuter à la fin
de la journée, avec une certaine probabilité. La question est donc de savoir à quelle vitesse
devront tourner les processeurs, et à quel point les services devront être répliqués sur les
machines. Nous exhibons des algorithmes d'approximation, prouvons leurs facteurs d'approximation
sur l'énergie consommée, et décrivons des heuristiques performantes.
}
\RRabstract{
  We consider allocation problems that arise in the context of service allocation in Clouds.  More
  specifically, we assume on the one part that each computing resource is associated to a capacity
  constraint, that can be chosen using Dynamic Voltage and Frequency Scaling (DVFS) method, and to a
  probability of failure.  On the other hand, we assume that the service runs as a set of
  independent instances of identical Virtual Machines.  Moreover, there exists a Service Level
  Agreement (SLA) between the Cloud provider and the client
that can be expressed as follows: the client comes with a minimal number of service instances which
must be alive at the end of the day, and the Cloud provider offers a list of pairs $(\mathit{price},
\mathit{compensation})$, this compensation being paid by the Cloud provider if it fails to keep
alive the required number of services. On the Cloud provider side, each pair corresponds actually to
a guaranteed success probability of fulfilling the constraint on the minimal number of instances.

In this context, given a minimal number of instances and a probability of success, the question for
the Cloud provider is to find the number of necessary resources, their clock frequency and an
allocation of the instances (possibly using replication) onto machines.  This solution should
satisfy all types of constraints during a given time period while minimizing the energy consumption
of used resources.  We consider two energy consumption models based on DVFS techniques, where the
clock frequency of physical resources can be changed. For each allocation problem and each energy
model, we prove deterministic approximation ratios on the consumed energy for algorithms that
provide guaranteed probability failures, as well as an efficient heuristic, whose energy ratio is
not guaranteed.
}
\RRmotcle{Cloud, fiabilité, approximation, économie d'énergie}
\RRkeyword{Cloud, reliability, approximation, energy savings}
\RRprojets{CEPAGE}
\RCBordeaux 

\begin{document}
\makeRR   

\section{Introduction}
\label{sec.intro}

\subsection{Reliability and Energy Savings in Cloud Computing}
This paper considers energy savings and reliability issues that arise when allocating instances of
an application consisting in a set of independent instances  running as Virtual Machines (VMs) onto
Physical Machines (PMs) in a Cloud computing platform. Cloud
Computing~\cite{zhang2010cloud,armbrust2009above,Buyya2009599,Greenberg:2008:CCR:1496091.1496103}
has emerged as a well-suited paradigm for service providing over the Internet. Using virtualization,
it is possible to run several Virtual Machines on top of a given Physical Machine. Since each VM
hosts its complete software stack (Operating System, Middleware, Application), it is moreover
possible to migrate VMs from a PM to another in order to dynamically balance the load.

In the static case, mapping VMs with heterogeneous computing demands onto PMs with (possibly
heterogeneous) capacities can be modeled as a multi-dimensional bin-packing problem. Indeed, in this
context, each physical machine is characterized by its computing capacity (\ie the number of flops
it can process during one time-unit), its memory capacity (\ie the number of different VMs that it
can handle simultaneously, given that each VM comes with its complete software stack) and its
failure rate (\ie the probability that the machine will fail during the next time period) and each
service comes with its requirement, in terms of CPU and memory demands, and reliability constraints.

In order to deal with capacity constraints in resource allocation problems, several sophisticated
techniques have been developed in order to optimally allocate VMs onto PMs, either to achieve good
load balancing~\cite{van2009sla,calheiros2009heuristic,christopherTPDS} or to minimize energy
consumption~\cite{berl2010energy,beloglazov2010energy}.  Most of the works in this domain have
therefore concentrated on designing offline~\cite{GareyJohnson} and
online~\cite{EpsteinvanStee07binpackresaumg,Hochbaum97} solutions of Bin Packing variants.

Reliability constraints have received much less attention in the context of Cloud computing, as
underlined by Walfredo Cirne in~\cite{CirneInvited}. Nevertheless, related questions have been
addressed in the context of more distributed and less reliable systems such as Peer-to-Peer
networks.  In such systems, efficient data sharing is complicated by erratic node failure,
unreliable network connectivity and limited bandwidth. Thus, data replication can be used to improve
both availability and response time and the question is to determine where to replicate data in
order to meet performance and availability requirements in large-scale
systems~\cite{replication_availability,replication_cirne_1,replication_datagrid,replication_DB,replication_cirne_2}.
Reliability issues have also been addressed by High Performance Computing community. Indeed,
recently, a lot of efforts has been done to build systems capable of reaching the Exaflop
performance~\cite{dongarra2009international,eesi} and such exascale systems are expected to gather
billions of processing units, thus increasing the importance of fault tolerance
issues~\cite{cappello2009fault}. Solutions for fault tolerance in Exascale systems are based on
replication strategies~\cite{ferreira2011evaluating} and rollback recovery relying on checkpointing
protocols~\cite{bougeret2011checkpointing,cappello2010checkpointing}.

This work is a follow-up of~\cite{beaumont:hal-00743524}, where the question of how to evaluate the
reliability of an allocation has been addressed and a set of deterministic and randomized heuristics
have been proposed. In this paper, we concentrate on energy savings issues and we propose proved
approximation algorithms. In order to minimize energy consumption, we assume that we can rely on
sophisticated mechanisms in order to fix the clock frequency of the PMs and we rely on techniques
such as DVFS (see~\cite{280894,pruhsTCS,pow3,pow3IPDPS,pow3ICPP}).
In this context, the capacity of the PM can be expressed as a function of the clock frequency. In
general, the probability of failure may itself depend on the clock frequency
(see for instance \cite{QiZA10a}).
(we will nevertheless not consider this case in this paper and we leave it for future works).

To assess precisely the specific complexity of energy minimization introduced by reliability
constraints in the context of services allocation in Clouds, we concentrate on a simple context,
that nevertheless captures the main difficulties.
First, we consider that the service running on the Cloud platform consists of a number of identical (in terms
  of requirements) and independent instances. Therefore, we do not consider the problems introduced
  by heterogeneity, that have already been considered (see for
  instance~\cite{calheiros2009heuristic,christopherTPDS}). Indeed, as soon as heterogeneity is
  considered, basic allocation problems are amenable to Bin Packing problem and are therefore
  intrinsically difficult. Then, we consider static allocation problems only, in the sense that our goal is to find the
  allocation that optimizes the reliability during a time period (say at the end of the day),
  instead of relying on VM migrations and creations to ensure that a minimal number of instances of
  each service is running whatever the machine failures.
Therefore, our work enables to assess precisely the complexity introduced by machine failures and
service reliability demands on energy minimization.

Throughout this paper, we assume that the characteristics of the applications and their requirements
(in terms of reliability in particular) have been negotiated between a client and the provider
through a Service Level Agreement (SLA). In the SLA, each service is characterized by its demand in
terms of processing capability (\ie the minimal number of instances of VMs that must be running
simultaneously) and in terms of reliability (\ie the maximal probability that the service will not
benefit from this number of instances at some point during the next time period). Equivalently, the
reliability requirement may be negociated through the payment of a fine by the cloud provider if it
fails to provide the required amount of resources.

In both cases, the goal, from the provider point of view, is to determine the cost of reliability,
since a higher reliability will induce more replication and therefore more energy consumption. Our
goal in this paper is to find allocations that minimize energy consumption while enforcing
reliability constraints and therefore to determine the cost of reliability. This cost of reliability
can then be directly translated into a set of (price, fine in case of SLA violation) offers by the
Cloud provider.

\subsection{Notations}
\label{sec.notations}

In this section, we introduce the notations that will be used throughout the paper. Our target cloud
platform is made of $m$ physical machines $\MM_1, \MM_2, \ldots, \MM_m$. As already noted, we assume
that Machine $\MM_j$ is able to handle the execution of $\capacity_j$ instances of service. We also
assume that we can rely on Dynamic Voltage Frequency Scaling (DVFS) mechanism in order to adapt
$\capacity_j$. The energy consumed by machine $\MM_j$ when running at capacity (speed
proportional to) $\capacity_j$ is given by $\energy = \estat(j) + \edyn(j)$, where $\edyn(j) = e_j
\capacity_j ^{\pen}$, with $\pen \geq 2$. This means that the energy consumed by machine $\MM_j$ can be
seen as the sum of a leakage term (paid as soon as the machine is switched on) and of a term that
depends (most of the works consider that $2 \leq \pen \leq 3$) on its running speed. For the sake of
simplicity, we will assume throughout this paper that any $\capacity_j$ can be achieved by Machine
$\MM_j$, as advocated in \cite{contmodes,juurParCo,contfreqIPDPS}.

On this Cloud platform, our goal is to run (all through a given time period, say a day, as defined
in the SLA) a service \service. \demand identical and independent
instances of service \service are required, and the instances run as
Virtual Machines. Several instances can therefore run concurrently and
independently on the same physical machine. We will denote by $\alloc_j$ the number of instances
running on $\MM_j$, that has to be smaller than $\capacity_j$. 
$\sum_j \alloc_{j}$ represents the overall number of running instances of \service. In general,
$\sum_j \alloc_{j}$ is larger than \demand since replication, \ie over-provisioning of services
is used in order to enforce reliability constraints.

More precisely, each machine $\MM_j$ comes with a failure rate $\fail_j$, that represents the
probability of failure of $\MM_j$ during the time period. During the time period, we will not
reallocate instances of service to physical machines but rather provision extra instances for the
service (replicas) that will actually be used if some machines fail. In practice, $\fail_j$ depends
on the clock frequency (although no clear consensus exist in the literature on how, see for instance
\cite{659037,hass1998mitigating,seifert2001historical}
) and therefore on $\capacity_j$. As said previously, we will assume for the results proved in this
paper that $\fail_j$ does not depend on $\capacity_j$.

We will denote by $\alive$ the set of running machines. In our model, at the end of the time period,
the machines are either up or completely down, so that the number of instances 
running on $\MM_j$ is $\alloc_{j}$ if $\MM_j \in \alive$ and 0 otherwise. Therefore, $\alivese =
\sum_{\MM_j \in \alive} \alloc_{j}$ denotes the overall number of running instances  at
the end of the time period and the service is running properly at the end of the time period if and
only if $\sum_{\MM_j \in \alive} \alloc_{j} \geq \demand$.

Of course, our goal is not that all services should run properly at the end of the time
period. Indeed, such a reliability cannot be achieved in practice since the probability that all
machines fail is clearly larger than 0 in our model. In general, as noted in a recent paper of the
NY Times
Data Centers usually over-provision resources (at the price of high energy consumption) in order to
(quasi-)avoid failures. In our model, we assume a more sustainable model, where the SLA defines the
reliability requirement \relia for the service (together with the penalty paid by the Cloud
Provider if \service does not run with at least \demand instances at the end of the
period). Therefore, the Cloud provider faces the following optimization problem:

$\textsc{\bf BestEnergy}(\nbm,\demand,\relia)$: Find the set \son of machines that are on and the clock frequency assigned to Machine $\MM_j$,
  represented by $\capacity_j$ and an allocation $\alloc$ of instances to machines
  $\MM_1, \MM_2, \ldots \MM_m$ 
such that 
(i) $\forall j \in \son, \alloc_{j} \leq \capacity_j$, (ii)
$\PP(\alivese \geq \demand) \geq 1- \relia$, \ie the probability that a least
  \demand instances of \service are running on alive machines after the time period is larger
  than the reliability requirement $1-\relia$, (iii)
and the overall energy consumption $\sum_{j \in \son} \estat(j) + e_j \capacity_j ^{\pen} $
is minimized.

\subsection{Methodology}
\label{sec.methodology}

Throughout the paper, we will rely on the same general approach. Through Section~\ref{minenns} to
Section~\ref{minen} and Section~\ref{algo}, we complicate the problem by considering more general
problem (from the problem of assigning a fractional number of instances onto a fixed number of
resources, to the problem of allocating integer number of instances onto a set of resources to be
determined).

In order to prove claimed approximation ratios, we will rely on the following techniques.

For the lower bounds, we prove that for a service, given the reliability constraints of the
  service and given failure probabilities of the machines, at least a given number of
  instances or at least a given level of energy is needed. These results are obtained through
  careful applications of Hoeffding Bounds~\cite{hoeffding}.

For the upper bounds, we concentrate on two special allocation schemes, namely \hhomo and
  \hstep. In a solution of \hhomo, for each service, we assign to every machine the same number of
  instances (that may be fractional or integral depending on the context), \ie $\forall i, \forall j
  \in \son, \all{i}{j} = \allhom{i}$. In a \hstep solution, we authorize one unit step, \ie $\forall
  i, \forall j \in \son, 0 \leq \all{i}{j} - \allhom{i} \leq 1$.  Using these two allocation
  schemes, we are able to derive theoretical bounds relying on Chernoff
  bounds~\cite{chernoff}. Moreover, the comparison with the lower bound shows that the quality of
  obtained solutions is reasonably high, especially in the case of energy minimization and even
  asymptotically optimal when the size of the platform or the overall volume of service instances to
  be handled, becomes arbitrarily large.

\subsection{Motivating example}
\label{sec.example}

In order to illustrate the objective functions that we consider throughout this paper and the notations, let us consider a service with a demand $\demand=20$ and a reliability request of $\relia=4.5 \cdot
10^{-6}$, that has to be mapped onto a cloud composed of $\nbm=10$ physical machines, whose failure
probability is $\fail=10^{-1}$. Figure~\ref{fig.motex} depicts the various kind of solutions that we
consider in this paper. In terms of minimizing the number of instances, the best solution consists in allocating 10 instances of the service to the first 2 machines and 5 instances to the 8 remaining machines. Therefore, the optimal solutions allocate a total of 60 instances, whereas 20 instances only are required at the end of the day, in order to satisfy reliability constraints. The shape of the optimal solution reflects the complexity of the problem. Indeed, Indeed, it has been proved
in~\cite{beaumont:hal-00743524} that even in the case with a single service and even if the
allocation is given, then estimating its reliability is $\# P'$-complete. The $\#P$ complexity class
has been introduced by Valiant~\cite{Valiant1979189} in order to classify the problems where the
goal is not to determine whether there exists a solution (captured by $\mathit{NP}$ completeness notion) but
rather to determine the number of solutions. In our context, the reliability of an allocation is
related to the number (weighted by their probability) of $\textsc{Alive}$ sets that lead to an
allocation where all service demands are satisfied. In this example, to check that the reliability is larger (in fact equal to) than $\relia$, we can observe that all configurations where at least 4 machines are alive are acceptable (since at least 20 instances are alive as soon as 4 machines are up), together with all configurations with 3 machines, as soon as a machine loaded with $10$ instances is
involved and the solution with only the first two machines alive. By counting the number of such valid configurations (weighted by their probability) leads to the reliability of the allocation. We can notice that the optimal solution involves $60$ instances against around $67$ for best fractional homogeneous
solution, and $64$ in the best step solution. Nevertheless, we will use fractional homogeneous and step solutions in order both to derive approximation algorithms and upper bounds on the number of required resources, and we will see that they are in general close to the optimal.

As far as energy minimization is concerned, we can notice that if we assume $\pen=2$, despite the bad
load-balancing among the machines in the optimal solution for the number of instances, this solution remains optimal. On the other hand, if $\pen=3$ for instance, then the best step solution consumes less energy than the solution minimizing the number of instances. We will prove in this paper that step and homogeneous fractional solutions are in fact asymptotically optimal when the overall demand, or the number of machines involved in the solution, becomes arbitrarily large.

\begin{figure}[h]
\begin{center}
\includegraphics[width=\textwidth]{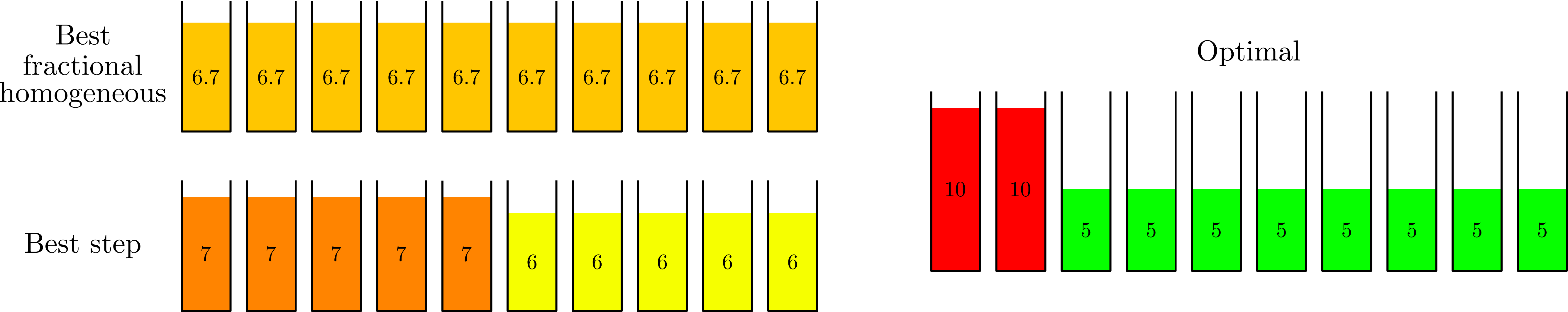}
\end{center}
\caption{Motivating example}
\label{contreexemple}
\label{fig.motex}
\end{figure}

\subsection{Outline of the Paper}
\label{sec.outline}

As we have noticed through the motivating example, \textsc{\bf BestEnergy} is in general difficult. 
Nevertheless, we prove in this paper that even when the allocation is to be determined, it is
possible to provide low-complexity deterministic approximation algorithms, that are even
asymptotically optimal when the sum of the demands becomes arbitrarily large. Another original
result that we prove in this paper is that minimizing the energy (relying on DVFS) induced by
replication is easier than minimizing the number of replicas, whereas in many contexts (see~\cite{ayp}
the non-linearity of energy consumption makes the optimization problems harder. In our context,
approximation ratio are smaller for energy minimization than for classical replica (that would
correspond to makespan or load balancing in other contexts) minimization.

To prove this result, we progressively come to the most general problem through the study of more
simple objective functions. Firstly, we address in Section~\ref{minenns} the case where we are given a
single service, where the set of machines that are switched on is given and where the number of
instances allocated to a machine is allowed to be fractional. Finding allocations such that the
number of instances we place is minimum is denoted as \minenns problem.  Then, we address in
Section~\ref{minen} the more general \minen problem. For \minen, the setting is the same except that
the number of participating machines is to be determined.  Finally, in Section~\ref{algo}, we tackle
the problem of designing more realistic solutions, where the number of instances on each machine must
be an integer. In \hhomo, all participating machines are allocated the same number of instances
whereas in \hstep, the number of instances allocated to a machine can differ by at most one (either
$a$ or $a+1$ for some value of $a$).

\section{Fractional \minenns}
\label{minenns}

\subsection{Lower bound}
\label{sec.lower}

Let us consider the case of a single service to be mapped onto a fixed number of machines when
the objective is to minimize the amount of resources necessary to enforce the conditions defined in the
SLA in terms of quantity (of alive instances at the end of the day) and reliability. The problem
comes into two flavours depending on the resources we want to optimize. Recall that \allhom{j} is
the number of instances of the service initially allocated to machine \mach{j}. In its physical
machines version, the optimization problem consists in minimizing the number of instances allocated to
the different machines, \ie $\sum_j \allhom{j}$. In its energy minimization version, we rely on DVFS
mechanism in order to adapt the voltage of a machine to the need of the instances allocated to
it. In general, energy consumption models assume that the energy dissipated by a processor running at
speed $s$ is proportional to $s^{\alpha}$. Therefore, the energy dissipated by a processor running
\allhom{j} instances will be proportional to $\allhom{j}^{\alpha}$ and the overall objective is to
minimize the overall dissipated energy, \ie $\sum_j \allhom{j}^{\alpha}$.

In order to find the lower bound, let us consider any allocation (where \allhom{j} is the number
service instances initially allocated to machine \mach{j}) and let us prove that if the
amount of resources is too small, then reliability constraints cannot be met. Recall that
\alivese[j] is the number of instances of the service that are alive on machine \mach{j} at the end of
the day. \alivese[j] is thus a random variable equal to \allhom{j} with a probability $1-\fail$
and to $0$ with a probability \fail.

Hence, the expected number of alive instances is given by $\expe{\alivese} = (1-\fail)
\sum_{j=1}^{\nbm} \alivese[j]$. Hoeffding inequality~(see~\cite{hoeffding}) says how much the number
of alive resources may differ from its expected value. In particular, for the lower bound, we will
use it in the following form, that bounds the chance of being lucky, \ie to find a correct
allocation with few instances. More precisely, it states that for all $t>0$:
\[ \pro \left( \alivese \geq \expe{\alivese} + t \right) \leq \exp \left( -2\frac{t^2}{\sum_{j=1}^{\nbm} \allhom{j} ^2} \right) . \]

Let us choose $t = \sqrt{ - \logn{1-\relia} \sum_{j=1}^{n} \allhom{j} ^2 /2} $, so that
$\exp \left( -2\frac{t^2}{\sum_{j=1}^{\nbm} \allhom{j} ^2} \right) = 1 - \relia$. Noting
$\lrd' = \frac{- \logn{1-\relia}}{2}$, we obtain that a
necessary condition on the \allhom{j}'s so that the reliability constraint is enforced is given by
$(1-\fail)\sum_{j=1}^{\nbm} \allhom{j} + \sqrt{ \lrd' \times \sum_{j=1}^{\nbm}  \allhom{j} ^2} \geq \demand. $

As stated in the introduction of this section, we are interested either in minimizing $\sum_j
\allhom{j}$ for resource use minimization, and $\sum_j \allhom{j}^{\alpha}$ for energy
minimization. To obtain lower bounds on these quantities in order to achieve quantitative (number of
alive instances) and qualitative (reliability constraints), we rely on Hoelder's inequality~cite{hoelder}, that
states that if $1/p+1/q=1$, then $\forall a_j, b_j \geq 0,~\sum a_j b_j \leq (\sum a_j^p)^{1/p},~
(\sum b_j^q)^{1/q} $. We assume in the following that $\pen>2$.

With $p=q=2, a_j=b_j=\allhom{j} $, we obtain $\sum \allhom{j}^2 \leq \left( \sum \allhom{j}
\right)^2$, so that \\$(1-\fail)\sum_{j=1}^{\nbm} \allhom{j} + \sqrt{ \lrd' \times \sum_{j=1}^{\nbm}
  \allhom{j} ^2} \leq \left( 1-\fail + \sqrt{ \lrd'} \right) \times
\sum_{j=1}^{\nbm} \allhom{j}.$
Hence a necessary condition in order to satisfy the constraints is given by $ \sum_{j=1}^{\nbm} \allhom{j} \geq \frac{\demand}{1-\fail + \sqrt{ \lrd'}}=\textsc{MinRep}.$\\
Therefore, any solution that satisfies quantitative and qualitative constraints
must allocate at least $\textsc{MinRep}$ instances, whatever the distribution of instances onto
machines is.

With $p=\alpha,~ 1/q=(1-1/\alpha),~ a_j=\allhom{j} $ and $b_j=1$, we obtain $\sum \allhom{j} \leq
\left( \sum \allhom{j}^{\alpha} \right)^{1/\alpha} \nbm^{1-1/{\alpha}}$. \\ Similarly, (remember
that $\alpha > 2$ so that $\alpha/2>1$), with $p=\alpha/2,~ 1/q=(1-2/\alpha),~ a_j=\allhom{j}^2 $
and $b_j=1$, we obtain $\sum \allhom{j}^2 \leq \left( \sum \allhom{j}^{\alpha} \right)^{2/\alpha}
\nbm^{1-2/{\alpha}}$, so that
\[ (1-\fail)\sum_{j=1}^{\nbm} \allhom{j} + \sqrt{ \lrd' \times \sum_{j=1}^{\nbm}
  \allhom{j} ^2} \leq \left( (1-\fail)\nbm^{1-1/{\alpha}} + \sqrt{ \lrd' }
\nbm^{1/2-1/{\alpha}} \right) \times \left( \sum \allhom{j}^{\alpha} \right)^{1/\alpha}. \]
Also, we can derive another necessary condition defined as
\[ \left( \sum \allhom{j}^{\alpha} \right) \geq \left( \frac{\demand}{(1-p)\nbm^{1-1/{\alpha}} +
  \sqrt{ \lrd' } \nbm^{1/2-1/{\alpha}} } \right)^{\alpha} = \textsc{MinEnergy}.\]
Therefore, any solution that satisfies quantitative and qualitative
constraints must consume at least \textsc{MinEnergy}, whatever the distribution of instances onto
machines is.

Note that both results hold true for $\pen=2$.

\subsection{Upper bound -- \hhomo}
\label{sec.upper}

\subsubsection{\minrep}
As explained above, in order to obtain upper bounds on the amount of necessary resources (either in terms of
number of instances or energy), it is enough to exhibit a valid solution (that satisfies the
constraints defined in the SLA). To achieve this, we will concentrate in
this part on homogeneous (fractional) solutions, with an equally-balanced allocation among all
machines (\ie $\forall j, \all{1}{j} = \alloc$).

An assignment is considered as failed when there are not enough instances of the service that are
running at the end of the day, hence $ \pf = \pro \left( \alivese \leq \demand \right)$. From the
homogeneous characteristics of the allocations, we derive that $\alivese = \alloc \times |\alive|$,
then $\pf = \pro\left( |\alive| \leq \frac{\demand}{\alloc} \right)$.
$|\alive|$ can be described as the sum of random independent
variables $\sum_{j=1}^{\nbm} X_j$, where, for all $j \in \inter{\nbm}$, $X_j$ depicts the fact that
machine $\mach{j}$ is alive at the end of the day ($X_j$ is equal to $1$ with probability $1-\fail$, and to
$0$ with probability $\fail$).

Hence, the expected value of $|\alive|$ is given by $\expe{|\alive|} = (1-\fail){\nbm}$. Chernoff
bound~(see~\cite{chernoff}) says how much the number of alive machines may differ from its
expected value. We use in this part Chernoff bounds rather than Hoeffding bounds because the
random variables take their value in $\{0,1\}$ instead of $\{0,\ldots,A\}$ and Chernoff bounds are
more accurate in this case. In particular, for the upper bound, we will use it in the following
form, that bounds the chance of being unlucky, \ie to fail having a correct allocation while allocating a large
number of instances. More specifically,
$ \pro\left( |\alive| \leq (1-\fail-\varepsilon)\nbm  \right) \leq e^{-2 \varepsilon^2 \nbm}. $
As we want to ensure that $\pf \leq \relia$, we choose $\varepsilon$ such that $e^{-2 \varepsilon^2
  \nbm}=\relia$, i.e. $\varepsilon = \sqrt{\lrd/\nbm}$ by noting $\lrd = \frac{- \logn{\relia}}{2}$.
Finally, we obtain a sufficient condition, so that the reliability constraint is fulfilled for the service
$ \alloc \nbm \geq \frac{\demand}{1- \fail - \sqrt{\frac{\lrd}{\nbm} }} = \textsc{MaxRep}.$

Therefore, it is possible to satisfy the SLA with at most \textsc{MaxRep} instances of the
service. Similarly, we can derive an upper bound of the energy needed to enforce the SLA. Indeed,
with the same value of $\alloc$, we obtain
$ E^{(\mathit{homo})} \geq \left( \frac{\demand}{(1- \fail) \nbm^{1-1/\alpha}
- \sqrt{\lrd} {\nbm}^{1/2-1/\alpha}} \right)^{\alpha}= \textsc{MaxEnergy}.$

\subsection{Comparison}

When minimizing the number of necessary instances to enforce the SLA, we obtain
$ \frac{\textsc{MaxRep}}{\textsc{MinRep}} = \frac{1-\fail + \sqrt{ \lrd'} }{1- \fail - \sqrt{\frac{\lrd}{\nbm} }}.$
For realistic values of the parameters, above approximation ratio is good (close to one), since both
$\sqrt{\lrd'}=\sqrt{ \frac{- \logn{1-\relia}}{2}}$ and
$\sqrt{\frac{\lrd}{\nbm}} = \sqrt{\frac{- \logn{\relia}}{2\nbm}}$ are small as soon as $\nbm$ is large.
Nevertheless, the ratio is not asymptotically optimal when $\nbm$ becomes large.

On the other hand, for energy minimization, we have
\[ \frac{\textsc{MaxEnergy}}{\textsc{MinEnergy}} = \left( \frac{(1-\fail)\nbm^{1-1/{\alpha}} + \sqrt{\lrd'}
\nbm^{1/2-1/{\alpha}} }{(1- \fail) \nbm^{1-1/\alpha} - \sqrt{\lrd}
  {\nbm}^{1/2-1/\alpha}} \right)^{\alpha} = \left( \frac{(1-\fail) + \sqrt{ \frac{\lrd'}{\nbm}} }
{(1- \fail) - \sqrt{\frac{\lrd}{\nbm}}} \right)^{\alpha},\]
 so that the ratio tends to 1
when $\nbm$ becomes arbitrarily large. This shows that for energy minimization, homogeneous
(fractional) solutions provide very good results when $\nbm$ is large enough. In the following section,
we prove that an allocation with a large dispersion (in a sense described precisely below) of the number of instances
allocated to the machines cannot achieve SLA constraints with optimal energy.

\subsection{Can optimal solutions be strongly heterogeneous ?}

Above results states that for the minimization of the number of instances and for the minimization
of the energy, homogeneous allocations provide good solutions. Nevertheless,
we know from the example depicted in Figure~\ref{contreexemple} that optimal solutions, for both the minimization of the number of instances and the minimization of
the energy are not always homogeneous.
In the case, of energy minimization, the dispersion of an allocation cannot be too large, as stated more formally in the following theorem.

\begin{theorem}
Let us consider a valid allocation $\allhom{j}$ whose energy is not larger than
$\textsc{MaxEnergy}$, the upper bound on the energy consumed by an homogeneous allocation. Then, if
$V' = \frac{ \sum (\allhom{j}^2)^{\alpha/2}}{\nbm} - \left( \frac{ \sum
  \allhom{j}^2}{\nbm} \right)^{\alpha/2}$
is used as the measure of dispersion of the $\allhom{j}$ (related to the $\alpha/2$-th moment of their square values), then
\[\nbm^{\alpha} V \leq \left( \frac{\demand}{1- \fail - \sqrt{\frac{\lrd}{\nbm} }}\right)^{\alpha} -
\left( \frac{\demand}{1-\fail + \sqrt{\lrd'} }\right)^{\alpha}.\]
\end{theorem}
\begin{proof}
Let us first introduce $V = \frac{ \sum \allhom{j}^{\alpha}}{\nbm} - \left( \frac{ \sum \allhom{j}}{\nbm} \right)^{\alpha}$. Then $V\geq V'$. Indeed, $V-V' = \left( \frac{ \sum
  \allhom{j}^2}{\nbm} \right)^{\alpha/2} - \left( \frac{ \sum \allhom{j}}{\nbm} \right)^{\alpha}$
that has the same sign as $\left( \frac{ \sum \allhom{j}^2}{\nbm} \right)^{1/2} - \left( \frac{ \sum
  \allhom{j}}{\nbm} \right)$ that is non-negative by application of Hoelder's inequality.

Moreover, we have seen that a necessary condition (see Section~\ref{sec.lower}) for allocation
$\allhom{j}$ to be valid is given by
$(1-\fail)\sum_{j=1}^{\nbm} \allhom{j} + \sqrt{ \lrd' \times \sum_{j=1}^{\nbm}
  \allhom{j} ^2} \geq \demand,$
what induces
$ (1-\fail) \left( \frac{\textsc{MinEnergy}}{\nbm} - V \right)^{1/\alpha} + \sqrt{ \lrd' \nbm }
 \left( \frac{\textsc{MinEnergy}}{\nbm} - V' \right)^{1/\alpha} \geq \demand $ 
and finally
$ V' < \frac{\textsc{MinEnergy}}{\nbm} - \left( \frac{\demand}{(1- \fail)\nbm - \sqrt{{\lrd}{\nbm} }} \right)^{\alpha}$
 or equivalently
$\nbm^{\alpha} V' \leq \left( \frac{\demand}{1- \fail - \sqrt{\frac{\lrd}{\nbm} }}\right)^{\alpha} -
\left( \frac{\demand}{1-\fail + \sqrt{ \lrd'} }\right)^{\alpha}.$
\end{proof}

\section{Fractional \minen}
\label{minen}

\subsection{Lower bound}

Let us know consider that the number of participating machines is to be determined. In this case, we need to take explicitly the leakage term into account (that was considered as constant in previous section since the number of switched on machine was fixed). In this case, given that $k \in \inter{\nbm}$, the goal is to minimize
\[ \ennl{k} = k \times \estat + k \times \left( \frac{\demand}{(1-\fail)k + \sqrt{\lrd' k}} \right) ^{\pen}.\]

Let $g$ be the function defined on $]0,+\infty[$ by $g(x) = g_t(x) / g_d^\pen(x)$. Let us prove that if
    $g_d$ is non-decreasing, concave, positive, and $g_t$ is non-increasing, convex and positive,
    then $g$ is convex. On the one hand, if $g_d$ fulfills its constraints, then $g_d^{-\pen}$ is
    non-increasing, convex and positive, and on the other hand, the product of two non-increasing,
    convex and positive is a convex function. 

Let us apply above lemma with $g_t(x) = x/x^{\pen/2}$
    (which is convex since $\pen>2$) and $g_d(x) = (1-\fail)\sqrt{x} + \sqrt{\lrd'}$, and deduce
    easily that \ennl{} is convex.

Therefore,\ennl{} admits a unique minimum on $[1,\nbm]$.  Since $\ennl{} \underset{0}{\rightarrow}
+\infty$ and $\ennl{} \underset{\infty}{\rightarrow} +\infty$, \dennl{} is null at some point in
$[0,+\infty[$, and let us define $\xmil$ such that $\dennl{\xmil} = 0$, \ie as
\begin{equation}
\label{eq.nlb}
 \estat + \left( \frac{\demand}{(1 - \fail)\xmil + \sqrt{\lrd' \xmil} } \right) ^\pen \times \left(
 -(\pen-1)(1-\fail) + \left( 1 - \frac{\pen}{2} \right) \sqrt{\frac{\lrd'}{\xmil}} \right) = 0.
\end{equation}
The minimum of function \ennl{} is reached on $[1,\nbm]$ for $\min(\max (\xmil , 1), \nbm)$.

We can also obtain a lower bound on the energy consumption if we restrict the search to integral number of
machines. Due to the convexity of \ennl{}, the minimum is achieved either at $\lceil \min(\max (\xmil , 1),\nbm) \rceil$ or $\lfloor \min(\max (\xmil,1),m) \rfloor$, so that
\[ \energy \leq \min\left(\ennl{ \lceil \min(\max (\xmil , 1),\nbm) \rceil}, \ennl{ \lfloor \min(\max (\xmil,1),m) \rfloor }\right). \]

\subsection{Upper bound -- \hhomo}
\label{sec.upper.var}

The energy consumption of an \hhomo solution on $k$ machines is given by
\[ \ennu{k} = k \times \estat + \frac{1}{k^{\pen-1}} \left( \frac{\demand}{(1-\fail) - \sqrt{\frac{\lrd}{k} }} \right) ^{\pen}. \]

Let us apply again above lemma with $g_t(x) = d^\pen / x^{\pen-1}$ and $g_d(x) = 1-\fail -
\sqrt{\frac{\lrd}{x}}$ to prove that \ennu{} is convex and consequently admits a unique minimum on
$[1,\nbm]$. Moreover, $\ennu{x} \underset{x \rightarrow \infty}{\longrightarrow} +\infty$ and $\ennu{x}
\underset{x \rightarrow 0}{\longrightarrow} +\infty$ so that we can uniquely define $\xmiu$ by
$\dennu{\xmiu}=0$, \ie
\begin{equation}
\label{eq.nhf}
 \estat = \left( \frac{\demand}{(1 - \fail)\xmiu - \sqrt{\lrd \xmiu} } \right) ^\pen \times \left(
 (\pen-1)(1-\fail) + \left( 1 - \frac{\pen}{2} \right) \sqrt{\frac{\lrd'}{\xmiu}} \right) .
\end{equation}

Therefore, we end up with the following upper bound for the energy
\[ \energy \geq \min\left(\ennl{ \lceil \min(\max (\xmil , 1),\nbm) \rceil}, \ennl{ \lfloor \min(\max (\xmil,1),m) \rfloor }\right). \]

\section{Algorithms for the Integral Case}
\label{algo}

In the service allocation problem in Clouds, demands represent a number of virtual machines that need
to be allocated onto physical machines. Therefore, the number of instances allocated to each machine has to be an integer, and we need to adapt above results in order to obtain valid allocation schemes. Moreover, the application of Chernoff bounds enables to find valid solutions (satisfying the reliability constraints) and to obtain theoretical upper bounds, but Chernoff bounds are in general too pessimistic, especially in the case of low number of machines. Hence, we derive in this section a few heuristics that return lower energy than those obtained in Section~\ref{sec.upper}.

\subsection{\minenns}

\subsubsection{Algorithms}

\paragraph{\glb}
In order to evaluate the performance of the heuristics, we rely on the lower bound proved in Section~\ref{sec.lower}. Since this lower bound is valid even among fractional solutions, it is \textit{a fortiori} valid for energy minimization for the integral problem.

\paragraph{\gth}
This algorithm builds a valid solution following the \hhomo policy. We have exhibited such a
solution in Section~\ref{sec.upper} on the fractional problem. In order to enforce the reliability constraint
while turning this solution into an integral one, we have to round the number of instances assigned
to each machine to the next integer, \ie $ \alloc = \left\lceil \frac{\demand}{\nbm(1-
  p - \sqrt{\frac{\lrd}{\nbm}) }} \right\rceil$, leading to an energy consumption of $ \edyn = \nbm \times \left\lceil \frac{\demand}{\nbm(1- p - \sqrt{\frac{\lrd}{\nbm}) }} \right\rceil ^\pen.$

\paragraph{\gbh}

This heuristic finds the best solution (\ie the one that minimizes the energy consumption) following
\hhomo policy. It can be decomposed into an off-line and an on-line phase; the former is executed once and
for all, while the latter is to be run for each reliability constraint.

In the off-line phase, we write a double-entry table, where a row is associated with a number of
machines $\nbm$ and a column corresponds to a reliability requirement $\relia$. The value of a cell
indicates the maximum number $\nbm'$ such that the probability of having $\nbm' \leq \nbm$ alive
machines among the $\nbm$ initial machines at the end of the day is not less than $1-\relia$. Those
values can be obtained thanks to a cumulative binomial distribution.

In the on-line phase, we perform a binary search on the machine capacity, so that we end up with a
valid solution minimizing the energy. Obviously, this solution is the solution that
minimizes the common capacity of the machines, and if the reliability constraint is fulfilled for a
given capacity, it is  \textit{a fortiori} true for a higher capacity. At each step, for a given capacity, we just
have to check, using the table, whether the number of alive instances is large enough.

\paragraph{\gbs}
This heuristic aims at relaxing the homogeneous constraint by finding the best solution on the
following form: there exists $\capacity$ such that the number of instances is either
$\capacity$ or $\capacity-1$. To achieve this goal, it first calls the previous \gbh heuristic that
returns $\capacity$. This ensures that an allocation of $\capacity$ instances per machine
leads to a valid solution, whereas if we allocate $\capacity-1$ instances to each machine, the
reliability constraint is violated. Then, we perform another binary search on the number of machines
that will hold $\capacity -1$ instances, instead of $\capacity$. The validity of a given allocation
is checked thanks to the dynamic programming algorithm described in~\cite{beaumont:hal-00743524}.

\subsubsection{Results}

In Figure~\ref{fig.ae}, we compare the performance of all heuristics under the following settings: $\fail = 10^{-2}$, $\demand = 500$, $\relia = 10^{-6}$, $\pen
= 2$ and $\nbm$ varies between 1 and 600. \glb is depicted in red, \gbs in pink, \gbh in blue, \gth in green and \gbs in red. We can clearly see the aliasing issue of \hhomo solutions on integral problem: during some
periods, it only increases the number of loaded machines without decreasing the overall capacity.

\hstep solutions almost solve completely this issue and softens the \gbh curve, still staying always
above. The ratio between the energy dissipated by \gbh and \glb is under $1.5$ as soon as $\nbm>25$.

\begin{figure}[h]
  \centering
    \includegraphics[width=\textwidth]{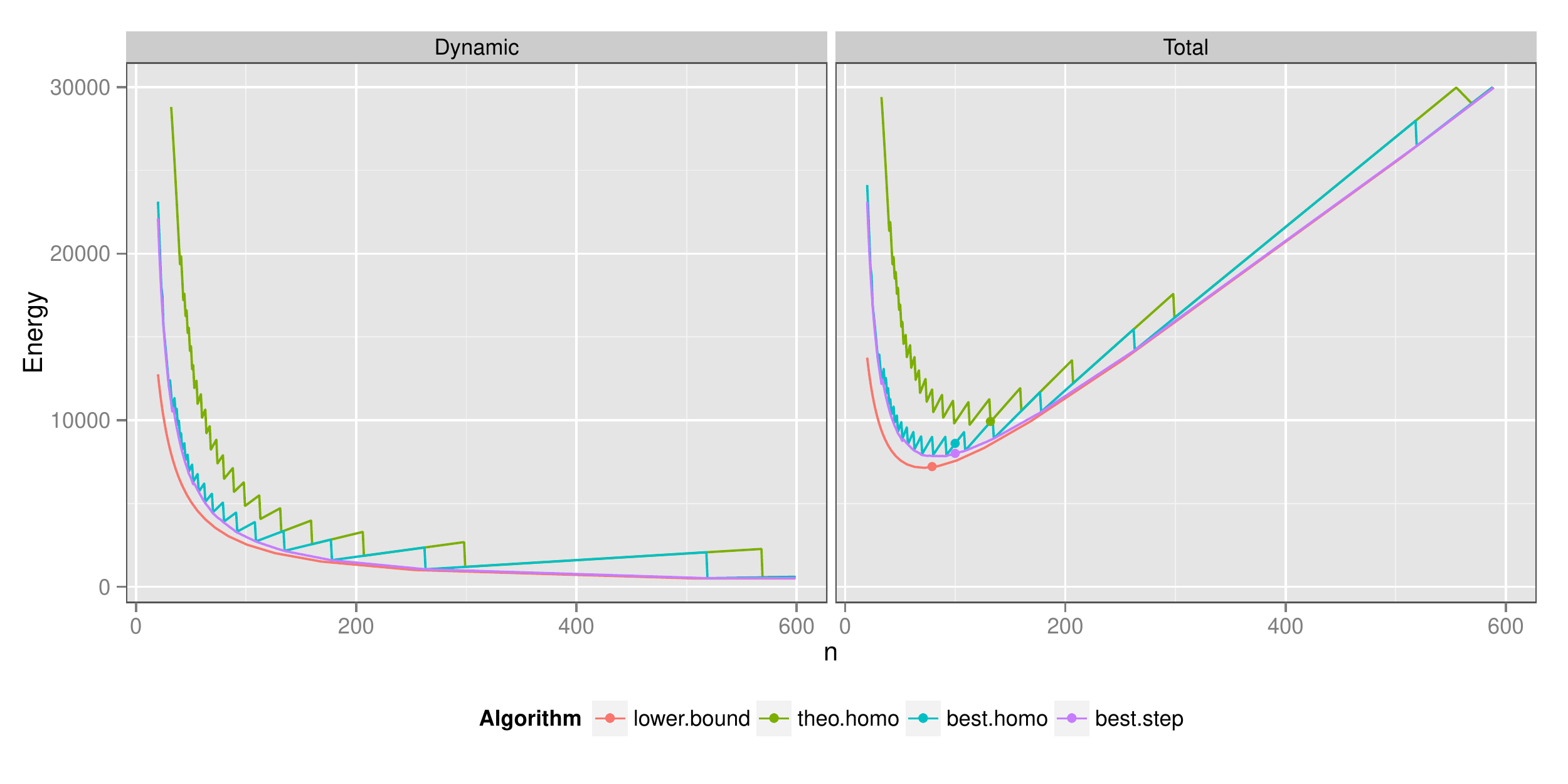}
  \caption{Simulation results for $\fail = 10^{-2}$, $\demand = 500$, $\relia = 10^{-6}$, $\pen= 2$, $1\leq\nbm\leq600$, $\estat=50$.}
  \label{fig.ae}
\end{figure}

\subsection{\minen}

When adding a non-zero static energy, all heuristics and bounds are such that the overall dissipated
energy tends to $+\infty$ if the number of machines tends to $0$ (because of the dynamic energy) or
to $+\infty$ (because of the static energy). There remains to find for each of them a close to the optimal
number of machines for each algorithm.

We have proved the convexity of the energy function returned by \glb. Thus, solving
Equation~\ref{eq.nlb} using binary search, is enough in order to obtain the optimal \nbm. When turning from fractional \hhomo solutions to integral
ones, convexity is lost and there is no easy way to find the optimal \nbm. Therefore, we try all possible number of
machines and keep the one that minimizes the consumed energy.

Concerning \gbh and \gbs, trying all possible number of machines would be too expensive, since
computing the consumed energy for a given \nbm is in general $\#P'$-complete. As the dynamic energy returned by
\gbh or \gbs lies between the dynamic energy given by the lower and upper bounds of the fractional
problem, the number of machines for \gbh and \gbs lies between the solutions of
Equation~\ref{eq.nlb} and Equation~\ref{eq.nhf}. Thus, we choose \nbm for \gbh and \gbs the mean of
previous solutions. The results for $\estat = 50$ are depicted in Figure~\ref{fig.ae}.

\section{Conclusion and Open Problems}
In this problem, we have considered approximation algorithms for minimizing
both the number of used resources and the dissipated energy in the context of service allocation
under reliability constraints on Clouds. For both optimization problems, we have given lower bounds
and have exhibited algorithms that achieve claimed reliability. In the case of energy
minimization, we have even been able to prove that proposed algorithm is asymptotically optimal when
the number of machines becomes arbitrarily large. Such a result is important
since it enables, for the Cloud provider point of view, to associate a price to reliability
(or to fix penalties in case of SLA violation). This work opens many perspectives. First, it seems
possible to improve, relying on different techniques, better approximation ratio in the case of low
number of resources. Then, the extension to several services is easy: all results can be generalized
except the lower bound on the energy consumption. Still we can use the lower bound, obtained for 
resource minimization and extend it to the energy minimization. At last, it would be
interesting to take explicitly into account the memory print of the services, so as to limit the
number of different services that a machine can handle. This would lead to different solution shapes,
by enforcing to limit the number of participating physical machines in the deployment of each
individual service.

\newpage

\bibliographystyle{abbrv}
\bibliography{biblio,biblio2}

\begin{thebibliography}{10}

\bibitem{armbrust2009above}
M.~Armbrust, A.~Fox, R.~Griffith, A.~Joseph, R.~Katz, A.~Konwinski, G.~Lee,
  D.~Patterson, A.~Rabkin, I.~Stoica, et~al.
\newblock {Above the clouds: A berkeley view of cloud computing}.
\newblock {\em EECS Department, University of California, Berkeley, Tech. Rep.
  UCB/EECS-2009-28}, 2009.

\bibitem{pow3IPDPS}
H.~Aydin and Q.~Yang.
\newblock Energy-aware partitioning for multiprocessor real-time systems.
\newblock In {\em Proceedings of the International Parallel and Distributed
  Processing Symposium (IPDPS)}, pages 113--121, 2003.

\bibitem{beaumont:hal-00743524}
O.~Beaumont, L.~Eyraud-Dubois, and H.~Larchev{\^e}que.
\newblock {Reliable Service Allocation in Clouds}.
\newblock In {\em {IPDPS 2013 - 27th IEEE International Parallel \& Distributed
  Processing Symposium}}, Boston, {\'E}tats-Unis, 2013.

\bibitem{christopherTPDS}
O.~Beaumont, L.~Eyraud-Dubois, H.~Rejeb, and C.~Thraves.
\newblock {Heterogeneous Resource Allocation under Degree Constraints}.
\newblock {\em IEEE Transactions on Parallel and Distributed Systems}, 2012.

\bibitem{beloglazov2010energy}
A.~Beloglazov and R.~Buyya.
\newblock {Energy efficient allocation of virtual machines in cloud data
  centers}.
\newblock In {\em 2010 10th IEEE/ACM International Conference on Cluster, Cloud
  and Grid Computing}, pages 577--578. IEEE, 2010.

\bibitem{ayp}
A.~Benoit, P.~Renaud-Goud, and Y.~Robert.
\newblock Power-aware replica placement and update strategies in tree networks.
\newblock In {\em IPDPS}, pages 2--13, 2011.

\bibitem{berl2010energy}
A.~Berl, E.~Gelenbe, M.~Di~Girolamo, G.~Giuliani, H.~De~Meer, M.~Dang, and
  K.~Pentikousis.
\newblock {Energy-efficient cloud computing}.
\newblock {\em The Computer Journal}, 53(7):1045, 2010.

\bibitem{bougeret2011checkpointing}
M.~Bougeret, H.~Casanova, M.~Rabie, Y.~Robert, and F.~Vivien.
\newblock Checkpointing strategies for parallel jobs.
\newblock In {\em High Performance Computing, Networking, Storage and Analysis
  (SC), 2011 International Conference for}, pages 1--11. IEEE, 2011.

\bibitem{659037}
S.~Buchner, M.~Baze, D.~Brown, D.~McMorrow, and J.~Melinger.
\newblock Comparison of error rates in combinational and sequential logic.
\newblock {\em Nuclear Science, IEEE Transactions on}, 44(6):2209 --2216, dec
  1997.

\bibitem{Buyya2009599}
R.~Buyya, C.~S. Yeo, S.~Venugopal, J.~Broberg, and I.~Brandic.
\newblock Cloud computing and emerging it platforms: Vision, hype, and reality
  for delivering computing as the 5th utility.
\newblock {\em Future Generation Computer Systems}, 25(6):599 -- 616, 2009.

\bibitem{calheiros2009heuristic}
R.~Calheiros, R.~Buyya, and C.~De~Rose.
\newblock {A heuristic for mapping virtual machines and links in emulation
  testbeds}.
\newblock In {\em 2009 International Conference on Parallel Processing}, pages
  518--525. IEEE, 2009.

\bibitem{cappello2009fault}
F.~Cappello.
\newblock Fault tolerance in petascale/exascale systems: Current knowledge,
  challenges and research opportunities.
\newblock {\em International Journal of High Performance Computing
  Applications}, 23(3):212--226, 2009.

\bibitem{cappello2010checkpointing}
F.~Cappello, H.~Casanova, and Y.~Robert.
\newblock Checkpointing vs. migration for post-petascale supercomputers.
\newblock {\em ICPP'2010}, 2010.

\bibitem{pow3}
A.~P. Chandrakasan and A.~Sinha.
\newblock Jouletrack: A web based tool for software energy profiling.
\newblock In {\em Design Automation Conference}, pages 220--225. IEEE CS Press,
  2001.

\bibitem{pow3ICPP}
J.-J. Chen and T.-W. Kuo.
\newblock Multiprocessor energy-efficient scheduling for real-time tasks.
\newblock In {\em Proceedings of International Conference on Parallel
  Processing (ICPP)}, pages 13--20. IEEE CS Press, 2005.

\bibitem{chernoff}
H.~Chernoff.
\newblock A measure of asymptotic efficiency for tests of a hypothesis based on
  the sum of observations.
\newblock {\em The Annals of Mathematical Statistics}, 23(4):493--507, 1952.

\bibitem{CirneInvited}
W.~Cirne.
\newblock Scheduling at google.
\newblock In {\em 16th Workshop on Job Scheduling Strategies for Parallel
  Processing (JSSPP), in conjunction with IPDPS 2012}, 2011.

\bibitem{replication_cirne_1}
D.~da~Silva, W.~Cirne, and F.~Brasileiro.
\newblock Trading cycles for information: Using replication to schedule
  bag-of-tasks applications on computational grids.
\newblock In H.~Kosch, L.~B\"osz\"orm\'enyi, and H.~Hellwagner, editors, {\em
  Euro-Par 2003 Parallel Processing}, volume 2790 of {\em Lecture Notes in
  Computer Science}, pages 169--180. Springer Berlin / Heidelberg, 2003.

\bibitem{dongarra2009international}
J.~Dongarra, P.~Beckman, P.~Aerts, F.~Cappello, T.~Lippert, S.~Matsuoka,
  P.~Messina, T.~Moore, R.~Stevens, A.~Trefethen, et~al.
\newblock The international exascale software project: a call to cooperative
  action by the global high-performance community.
\newblock {\em International Journal of High Performance Computing
  Applications}, 23(4):309--322, 2009.

\bibitem{eesi}
Eesi, "the european exascale software initiative", 2011.
\newblock \url{http://www.eesi-project.eu/pages/menu/homepage.php}.

\bibitem{EpsteinvanStee07binpackresaumg}
L.~Epstein and R.~van Stee.
\newblock Online bin packing with resource augmentation.
\newblock {\em Discrete Optimization}, 4(3-4):322--333, 2007.

\bibitem{ferreira2011evaluating}
K.~Ferreira, J.~Stearley, J.~Laros~III, R.~Oldfield, K.~Pedretti,
  R.~Brightwell, R.~Riesen, P.~Bridges, and D.~Arnold.
\newblock Evaluating the viability of process replication reliability for
  exascale systems.
\newblock In {\em Proceedings of 2011 International Conference for High
  Performance Computing, Networking, Storage and Analysis}, page~44. ACM, 2011.

\bibitem{GareyJohnson}
M.~R. Garey and D.~S. Johnson.
\newblock {\em Computers and Intractability, a Guide to the Theory of
  {NP}-Completeness}.
\newblock W.\ H. Freeman and Company, 1979.

\bibitem{Greenberg:2008:CCR:1496091.1496103}
A.~Greenberg, J.~Hamilton, D.~A. Maltz, and P.~Patel.
\newblock The cost of a cloud: research problems in data center networks.
\newblock {\em SIGCOMM Comput. Commun. Rev.}, 39(1):68--73, Dec. 2008.

\bibitem{hass1998mitigating}
K.~J. Hass, J.~W. Gambles, B.~Walker, and M.~Zampaglione.
\newblock Mitigating single event upsets from combinational logic.
\newblock In {\em 7th NASA Symposium on VLSI Design}, volume~4, pages 1--4,
  1998.

\bibitem{Hochbaum97}
D.~Hochbaum.
\newblock {\em Approximation Algorithms for NP-hard Problems}.
\newblock PWS Publishing Company, 1997.

\bibitem{hoeffding}
W.~Hoeffding.
\newblock Probability inequalities for sums of bounded random variables.
\newblock {\em Journal of the American Statistical Association},
  58(301):13--30, 1963.

\bibitem{replication_DB}
H.-I. Hsiao and D.~J. Dewitt.
\newblock A performance study of three high availability data replication
  strategies.
\newblock {\em Distributed and Parallel Databases}, 1:53--79, 1993.
\newblock 10.1007/BF01277520.

\bibitem{280894}
T.~Ishihara and H.~Yasuura.
\newblock Voltage scheduling problem for dynamically variable voltage
  processors.
\newblock In {\em Proceedings of International Symposium on Low Power
  Electronics and Design (ISLPED)}, pages 197--202. ACM Press, 1998.

\bibitem{contmodes}
T.~Ishihara and H.~Yasuura.
\newblock Voltage scheduling problem for dynamically variable voltage
  processors.
\newblock In {\em Proceedings of International Symposium on Low Power
  Electronics and Design (ISLPED)}, pages 197--202, New York, NY, USA, 1998.
  ACM Press.

\bibitem{juurParCo}
P.~Langen and B.~Juurlink.
\newblock Leakage-aware multiprocessor scheduling.
\newblock {\em Journal of Signal Processing Systems}, 57(1):73--88, 2009.

\bibitem{replication_datagrid}
M.~Lei, S.~V. Vrbsky, and X.~Hong.
\newblock An on-line replication strategy to increase availability in data
  grids.
\newblock {\em Future Generation Computer Systems}, 24(2):85 -- 98, 2008.

\bibitem{contfreqIPDPS}
R.~Mishra, N.~Rastogi, D.~Zhu, D.~Moss\'{e}, and R.~Melhem.
\newblock Energy aware scheduling for distributed real-time systems.
\newblock In {\em Proceedings of the International Parallel and Distributed
  Processing Symposium (IPDPS)}, pages 21--29, 2003.

\bibitem{pruhsTCS}
K.~Pruhs, R.~van Stee, and P.~Uthaisombut.
\newblock Speed scaling of tasks with precedence constraints.
\newblock {\em Theory of Computing Systems}, 43:67--80, 2008.

\bibitem{QiZA10a}
X.~Qi, D.~Zhu, and H.~Aydin.
\newblock Global reliability-aware power management for multiprocessor
  real-time systems.
\newblock In {\em RTCSA}, pages 183--192, 2010.

\bibitem{replication_availability}
K.~Ranganathan, A.~Iamnitchi, and I.~Foster.
\newblock Improving data availability through dynamic model-driven replication
  in large peer-to-peer communities.
\newblock In {\em Cluster Computing and the Grid, 2002. 2nd IEEE/ACM
  International Symposium on}, page 376, may 2002.

\bibitem{replication_cirne_2}
E.~Santos-Neto, W.~Cirne, F.~Brasileiro, and A.~Lima.
\newblock Exploiting replication and data reuse to efficiently schedule
  data-intensive applications on grids.
\newblock In D.~Feitelson, L.~Rudolph, and U.~Schwiegelshohn, editors, {\em Job
  Scheduling Strategies for Parallel Processing}, volume 3277 of {\em Lecture
  Notes in Computer Science}, pages 54--103. Springer Berlin / Heidelberg,
  2005.

\bibitem{seifert2001historical}
N.~Seifert, D.~Moyer, N.~Leland, and R.~Hokinson.
\newblock Historical trend in alpha-particle induced soft error rates of the
  alpha< sup> tm</sup> microprocessor.
\newblock In {\em Reliability Physics Symposium, 2001. Proceedings. 39th
  Annual. 2001 IEEE International}, pages 259--265. IEEE, 2001.

\bibitem{Valiant1979189}
L.~Valiant.
\newblock The complexity of computing the permanent.
\newblock {\em Theoretical Computer Science}, 8(2):189 -- 201, 1979.

\bibitem{van2009sla}
H.~Van, F.~Tran, and J.~Menaud.
\newblock {SLA-aware virtual resource management for cloud infrastructures}.
\newblock In {\em IEEE Ninth International Conference on Computer and
  Information Technology}, pages 357--362. IEEE, 2009.

\bibitem{zhang2010cloud}
Q.~Zhang, L.~Cheng, and R.~Boutaba.
\newblock {Cloud computing: state-of-the-art and research challenges}.
\newblock {\em Journal of Internet Services and Applications}, 1(1):7--18,
  2010.

\end{thebibliography}

\newpage
\tableofcontents

\end{document}